\documentclass[runningheads]{llncs}
\usepackage{amsmath}
\usepackage{amsfonts}
\usepackage{amssymb}
\usepackage{verbatim}
\usepackage[pdftex]{graphicx}
  \graphicspath{{./pics/}}
  \DeclareGraphicsExtensions{.pdf}
\usepackage{tikz}
\usetikzlibrary{arrows,automata,shapes}
\usetikzlibrary{positioning}
  
\newcommand{\eps}{\varepsilon}
\newcommand{\R}{$\mathcal{R}$}
\newcommand{\J}{$\mathcal{J}$}
\newcommand{\makeset}[2]{\ensuremath{ \{ #1 \: | \: #2 \} }}
\renewcommand{\sc}{\mathrm{sc}}

\spnewtheorem{fact}{Fact}{\bf}{\it}

\begin{document}
\title{On the State Complexity of the Reverse of \R-~and \J-trivial Regular Languages}
\titlerunning{State Complexity of the Reverse of \R- and \J-trivial Regular Languages}
\authorrunning{G. Jir\'{a}skov\'{a}, T. Masopust}
\author{
  Galina Jir\'{a}skov\'{a}\,\inst{1,}%
  \thanks{Research supported by VEGA grant 2/0183/11 and by grant APVV-0035-10.}
  \and
  Tom\'{a}\v{s} Masopust\,\inst{2,3,}%
  \thanks{Research supported by GA{\v C}R grant P202/11/P028 and by RVO: 67985840.}%
  }
\institute{
  Mathematical Institute, Slovak Academy of Sciences\\
  Gre{\v s}{\' a}kova 6, 040 01 Ko\v{s}ice, Slovak Republic\\
  \email{jiraskov@saske.sk}
  \and
  Institute of Mathematics, Academy of Sciences of the Czech Republic\\
  {\v Z}i{\v z}kova 22, 616 62 Brno, Czech Republic\\
  \email{masopust@math.cas.cz}
  \and
  Institute for Computer Science, University of Bayreuth
}

\maketitle
\begin{abstract}
  The tight upper bound on the state complexity of the reverse of \R-trivial and \J-trivial regular languages of the state complexity $n$ is $2^{n-1}$. The witness is ternary for \R-trivial regular languages and $(n-1)$-ary for \J-trivial regular languages. In this paper, we prove that the bound can be met neither by a binary \R-trivial regular language nor by a \J-trivial regular language over an $(n-2)$-element alphabet. We provide a characterization of tight bounds for \R-trivial regular languages depending on the state complexity of the language and the size of its alphabet. We show the tight bound for \J-trivial regular languages over an $(n-2)$-element alphabet and a few tight bounds for binary \J-trivial regular languages. The case of \J-trivial regular languages over an $(n-k)$-element alphabet, for $2\le k\le n-3$, is open.
\end{abstract}

\section{Introduction}
  Regular languages of simple forms play an important role in mathematics and computer science. The reader is referred to, e.g.,~\cite{Bojanczyk:2012,2013arXiv1303.0966C,Rogers:2007} for a few applications of \J-trivial (piecewise testable) languages. The aim of this paper is to investigate the state complexity of the reverse of two such language classes, namely of \R-trivial and \J-trivial regular languages. 
  
  For a regular language, the state complexity is the number of states of its minimal automaton representation. The reverse of an automaton or of a language is a classical operation whose state complexity is exponential in the worst case. There exist binary witness languages of the state complexity $n$ with the reverse of the state complexity $2^{n}$, see~\cite{le81,yzs94}. This even holds true for union-free regular languages defined by regular expressions without the union operation~\cite{ijfcsJiraskovaM11}.

  As mentioned above, we consider languages defined by Green's equivalence relations, namely \R-trivial and \J-trivial regular languages. Let $M$ be a monoid and $s$ and $t$ be two elements of $M$. Green's relations $\mathcal{L}$, \R, \J, and $\mathcal{H}$ on $M$ are defined so that
    $(s,t)\in\mathcal{L}$ if and only if $M\cdot s = M\cdot t$,
    $(s,t)\in\mathcal{R}$ if and only if $s\cdot M = t\cdot M$,
    $(s,t)\in\mathcal{J}$ if and only if $M\cdot s\cdot M = M\cdot t\cdot M$, and
    $\mathcal{H}=\mathcal{L}\cap\mathcal{R}$.
  For $\rho\in\{\mathcal{L},\mathcal{R},\mathcal{J},\mathcal{H}\}$, $M$ is $\rho$\mbox{-}trivial if $(s,t)\in\rho$ implies $s = t$, for all $s,t$ in $M$. A language is {\em $\rho$\mbox{-}trivial} if its syntactic monoid is $\rho$-trivial. Note that $\mathcal{H}$-trivial regular languages coincide with {\em star-free} languages~\cite[Chapter~11]{lawson2003finite} and that $\mathcal{L}$-trivial, \R-trivial and \J-trivial regular languages are all star-free. Moreover, \J-trivial regular languages are both $\mathcal{L}$-trivial and \R-trivial.

  Equivalently, a regular language is \R-trivial if and only if it is a finite union of languages of the form $\Sigma_1^* a_1 \Sigma_2^* a_2 \Sigma_3^* \cdots \Sigma_k^* a_k \Sigma^*$, where $k\ge 0$, $a_i\in \Sigma$, and $\Sigma_i\subseteq \Sigma \setminus \{a_i\}$, or if and only if it is accepted by a partially ordered minimal DFA~\cite{BrzozowskiF80}. Similarly, a regular language is \J-trivial (or {\em piecewise testable}) if and only if it is a finite boolean combination of languages of the form $\Sigma^*a_1\Sigma^*a_2\Sigma^*\ldots \Sigma^*a_k\Sigma^*$, where $k\ge 0$ and $a_i\in \Sigma$, or if and only if the minimal DFAs for both the language and the reverse of the language are partially ordered~\cite{Simon1972,Simon1975}. Other automata representations of these languages can be found, e.g., in~\cite{LauserCIAA} and the literature therein. Stern~\cite{Stern85a} suggested a polynomial algorithm 
  of order $O(n^5)$ in the number of states and transitions of the minimal DFA 
  to decide whether a regular language is \J-trivial. Trahtman~\cite{Trahtman2001} recently improved this result to a quadratic algorithm. 
  
  In~\cite{ciaa2012}, we have shown that the upper bound on the state complexity of the reverse of \R-trivial and \J-trivial regular languages is $2^{n-1}$ for languages of the state complexity $n$. We have also shown that this bound can be met by a ternary \R-trivial regular language and conjectured that an ($n-1$)-element alphabet is sufficient for \J-trivial regular languages of the state complexity $n$ to meet the upper bound, which was later proved in~\cite{BrArXiv12}. In this paper, we prove the optimality of the size of these alphabets. Namely, we prove that the bound on the state complexity of the reverse can be met neither by a binary \R-trivial regular language (Lemma~\ref{lem:binaryRtriv_upper}) nor by a \J-trivial regular language over an $(n-2)$-element alphabet (Theorem~\ref{thm:boundJtriv}). As a result, we provide a complete characterization of tight upper bounds for \R-trivial regular languages depending on the state complexity of the language and the size of its alphabet (Theorem~\ref{thm1}). Finally, we prove a tight upper bound for \J-trivial regular languages over $(n-2)$-element alphabets (Theorem~\ref{thm3}) and several tight bounds for binary \J-trivial regular languages (Table~\ref{tab}). The case of \J-trivial regular languages over $(n-k)$-element alphabets, for $2\le k\le n-3$, is left open.

\section{Preliminaries and Definitions}
  We assume that the reader is familiar with automata and formal language theory. The cardinality of a set $A$ is denoted by $|A|$, and the powerset of $A$ is denoted by $2^A$. An alphabet is a finite nonempty set. The free monoid generated by an alphabet $\Sigma$ is denoted by $\Sigma^*$. A string over $\Sigma$ is any element of $\Sigma^*$, and the empty string is denoted by $\eps$.

  A {\em nondeterministic finite automaton\/} (NFA) is a 5-tuple $M = (Q,\Sigma,\delta,Q_0,F)$, where $Q$ is the finite nonempty set of states, $\Sigma$ is the input alphabet, $Q_0\subseteq Q$ is the set of initial states, $F\subseteq Q$ is the set of accepting states, and $\delta:Q\times\Sigma\to 2^Q$ is the transition function that can be extended to the domain $2^Q\times\Sigma^*$. The language {\em accepted\/} by $M$ is the set $L(M) = \{w\in\Sigma^* \mid \delta(Q_0, w) \cap F \neq\emptyset\}$. The NFA $M$ is {\em deterministic\/} (DFA) if $|Q_0|=1$ and $|\delta(q,a)|=1$ for every $q$ in $Q$ and $a$ in $\Sigma$. In this case we identify singleton sets with their elements and simply write $q$ instead of $\{q\}$. Moreover, the transition function $\delta$ is a total map from $Q\times\Sigma$ to $Q$ that can be extended to the domain $Q\times\Sigma^*$. Two states of a DFA are {\em distinguishable\/} if there exists a string $w$ that is accepted from one of them and rejected from the other; otherwise they are {\em equivalent}. A DFA is {\em minimal\/} if all its states are reachable and pairwise distinguishable. A non-accepting state $d\in Q$ such that $\delta(d,a)=d$, for all $a$ in $\Sigma$, is called a {\em dead\/} state.
  
  The \emph{state complexity} of a regular language $L$, denoted by $\sc(L)$, is the number of states in the minimal DFA accepting the language $L$.

  The {\em subset automaton\/} of an NFA $M=(Q,\Sigma,\delta,Q_0,F)$ is the DFA  $M'=(2^Q,\Sigma,\delta',Q_0,F')$ constructed by the standard subset construction.
  
  Let $M=(Q,\Sigma,\delta,Q_0,F)$ be a DFA. The reachability relation $\preceq$ on the states of $M$ is defined by $p\preceq q$ if there exists a string $w$ in $\Sigma^*$ such that $\delta(p,w)=q$. The DFA $M$ is {\em partially ordered\/} if the reachability relation $\preceq$ is a partial order. For two states $p$ and $q$ of $M$, we write $p \prec q$ if $p\preceq q$ and $p\ne q$. A state $p$ is {\em maximal\/} if there is no state $q$ such that $p\prec q$. 

  The \emph{reverse $w^R$ of a string $w$\/} is defined by $\eps^R=\eps$ and $(va)^R=av^R$, for $v$ in $\Sigma^*$ and $a$ in $\Sigma$. The \emph{reverse of a language\/} $L$ is the language $L^R=\{w^R\mid w\in L\}$. The \emph{reverse of a DFA\/} $M$ is the NFA $M^R$ obtained from $M$ by reversing all transitions and swapping the role of initial and accepting states. 
  The following result says that there are no equivalent states in the subset automaton of the reverse of a minimal DFA. We use this fact in the paper when proving the tightness of upper bounds. By this fact, it is sufficient to show that the corresponding number of states is reachable in the subset automaton since the distinguishability always holds.
  \begin{fact}[\cite{br63}]\label{le:equiv}
    All states of the subset automaton corresponding to  the reverse of a minimal DFA are pairwise distinguishable.
    \qed
  \end{fact}

  In what follows we implicitly use the characterization that a regular language is \R-trivial if and only if it is accepted by a minimal partially ordered DFA and that it is \J-trivial if and only if both the language and its reverse are accepted by minimal partially ordered DFAs. This characterization immediately implies that \J-trivial regular languages are closed under reverse. However, \R-trivial regular languages are not closed under reverse since not all \R-trivial regular languages are \J-trivial. For instance, the \R-trivial regular language of Fig.~\ref{fi:binaryRtriv} is not \J-trivial, hence the minimal DFA for its reverse is not partially ordered.

  The following lemma shows that in some cases we do not need to distinguish between DFAs with and without dead state. In particular, we can get a result for DFAs without a dead state immediately from the analogous result for DFAs with a dead state or vice versa.\footnote{We are grateful to an anonymous referee for pointing out this observation.} 
  \begin{lemma}\label{lem:complement}
    Let $L$ be a regular language. Then $\sc(L)=\sc(L^c)$, where $L^c$ denotes the complement of $L$. 
    In particular, we have $\sc(L^R)=\sc((L^c)^R)$.
  \end{lemma}
  \begin{proof}
    Let $M$ be a minimal DFA accepting $L$. Then $M^c$ constructed from $M$ by swapping accepting and non-accepting states is a minimal DFA accepting $L^c$. The second part now follows by the observation that $(L^R)^c=(L^c)^R$. 
    Indeed, $w\in (L^R)^c$ if and only if $w\notin L^R$ if and only if $w^R\notin L$ if and only if $w^R\in L^c$ if and only if $w\in (L^c)^R$.
    \qed
  \end{proof}

  Let $M$ be a DFA with a dead state reaching the upper bound on the reverse. This lemma says that if the complement of $M$ does not have a dead state, the same result can be reached by DFAs without a dead state. Indeed, the complement of $M$ reaches the bound. However, Table~\ref{tab} demonstrates that there are cases where this technique fails because both the DFA and its complement have a dead state. 
  
  Immediate consequences of this lemma combined with the known results are formulated below.
  \begin{corollary}\label{cor:bound}$ $
    \begin{itemize}
      \item[(i)] There exist ternary \R-trivial regular languages $L_1$ and $L_2$ whose automaton representation has and does not have a dead state, respectively, with $\sc(L_1)=\sc(L_2)=n$ and $\sc(L_1^R)=\sc(L_2^R)=2^{n-1}$. 
      \item[(ii)] There exist \J-trivial regular languages $L_1$ and $L_2$ over an alphabet $\Sigma$ with $|\Sigma|\ge n-1$ whose automaton representation has and does not have a dead state, respectively, with $\sc(L_1)=\sc(L_2)=n$ and $\sc(L_1^R)=\sc(L_2^R)=2^{n-1}$.
    \end{itemize}
  \end{corollary}
  \begin{proof}
    Using Lemma~\ref{lem:complement},
    (i) follows from~\cite[Lemma~3, p.~232]{ciaa2012} since the automaton used there has a dead state and its complement does not,
    while (ii) follows from the automaton used in~\cite[Theorem~5, p.~15]{BrArXiv12}.
  \qed
  \end{proof}

\section{\R-trivial regular languages}
  Recall that the state complexity of the reverse for \R-trivial regular languages with the state complexity $n$ is $2^{n-1}$ and there exists a ternary witness language meeting the bound~\cite{ciaa2012}. 
  We now prove that the ternary alphabet is optimal,
  that is, the bound
  cannot be met by any binary \R-trivial regular language. 

  \begin{lemma}\label{lem:binaryRtriv_upper}
    Let $L$ be a binary \R-trivial regular language with $\sc(L)=n$, where $n\ge2$.
    Then $\sc(L^R)\le 2^{n-2}+n-1$. 
  \end{lemma}
  \begin{proof}
    Let $M=(\{1,\ldots,n\},\{a,b\},\delta,1,F)$ be a minimal partially ordered DFA
    with $n$ states such that $i\preceq j$ implies $i\le j$. 
    Let $M'$ denote the subset automaton of the NFA $M^R$.
    We show that $M'$ has at most $n-1$ reachable states that do not contain $n-1$. 
    By Lemma~\ref{lem:complement}, we can assume that state $n$ of $M$ is accepting,    
    otherwise we take the complement of $M$. 
    Then there are three cases in $M$ between states $n-1$ and $n$:
    (i)   state $n-1$ has self-loops under both letters $a$ and $b$,
    (ii)  both letters $a,b$ go from state $n-1$ to state $n$, or
    (iii) without loss of generality, the transition under $b$ goes from $n-1$ to $n$ and $a$ is a self-loop in state $n-1$.

    In the first case, 
    states $n$ and $n-1$ have self-loops under both letters in $M$.
    As $n-1$ is non-accepting (otherwise equivalent to $n$), $n$ appears in all and $n-1$ in no reachable states of $M'$.
    This gives at most $2^{n-2}$ reachable states in $M'$.
    
    In the second case,
    no sets without state $n-1$ are reachable in $M'$, except for $F$, 
    because state $n$ appears in all reachable states of $M'$ and 
    any transition of $M'$ generates state $n-1$ into the next state.
    Thus, expect for the initial state $F$ of $M'$, every reachable state of $M'$ contains both $n$ and $n-1$.
    Hence, the upper bound is at most $2^{n-2}+1$.

    In the third case, all subsets not containing state $n-1$  
    must be reachable in $M'$ 
    by strings in $a^*$. 
    We prove that at most $n-1$ such sets are reachable in $M'$.
    To this aim, it is sufficient to show that $F\cdot a^{n-1} = F\cdot a^{n-2}$,
    where $\cdot$ denotes the transition function of the subset automaton $M'$.
    The subautomaton of $M$, defined by restricting to the alphabet $\{a\}$,
    is a disjoint union of trees $T_q$ where $\delta(q,a)=q$ and $T_q$
    consists of all states that can reach $q$ by a string in $a^*$;
    see Fig.~\ref{fig:trees} for illustration.
    \begin{figure}[t]
      \centering
      \begin{tikzpicture}[->,>=stealth,shorten >=1pt,auto,node distance=1.8cm,
        state/.style={circle,minimum size=6mm, very thin,draw=black,initial text=}]
        \node[state,initial] (0) {$0$};
        \node[state] (1) [above right of=0] {$1$};
        \node[state] (2) [right of=0] {$2$};
        \node[state] (3) [right of=1] {$3$};
        \node[state] (5) [below right of=2] {$5$};
        \node[state] (6) [above right of=5] {$6$};
        \node[state] (7) [right of=5] {$7$};
        \node[state] (8) [right of=6] {$8$};
        \node[state] (4) [above of=8,node distance=1.4cm] {$4$};
        \node[state,accepting] (9) [right of=8] {$9$};
        \path
          (0) edge node {$a$} (1)
          (1) edge node {$a$} (3)
          (2) edge node {$a$} (3)
          (3) edge node {$a$} (4)
          (4) edge[loop above] node {$a,b$} (4)
          (5) edge node {$a$} (6)
          (6) edge node {$a$} (8)
          (8) edge[loop above] node {$a$} (8)
          (7) edge node {$a,b$} (8)
          (9) edge[loop above] node {$a,b$} (9)
          (0) edge[dotted] node {$b$} (2)
          (1) edge[bend left=25,dotted] node {$b$} (4)
          (2) edge[dotted] node {$b$} (5)
          (3) edge[dotted] node {$b$} (8)
          (5) edge[dotted] node {$b$} (7)
          (6) edge[dotted] node {$b$} (4)
          (8) edge[dotted] node {$b$} (9) ;
      \end{tikzpicture}
      \caption{There are three trees, namely $T_4=\{0,1,2,3,4\}$, $T_8=\{5,6,7,8\}$, and $T_9=\{9\}$; $b$-transitions are dotted.}
      \label{fig:trees}
    \end{figure}
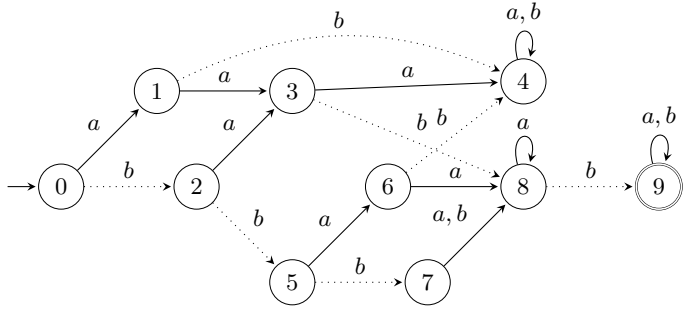
    Let $k$ be the depth of $T_q$, and let $F' = F\cap T_q$.
    If $q \in F'$, then $F'\cdot a^{k} =T_q$. 
    If $q \notin F'$, then $F'\cdot  a^{k}=\emptyset$.
    In both cases, $F'\cdot a^{k} =F'\cdot a^{k+1}$.
    Now $F \cdot a^m$ is a disjoint union of such $F'\cdot a^m$. 
    By the assumption, all trees are of depth at most $n-2$;
    recall that there is no $a$-transition from $n-1$ to~$n$.
    Hence $ F\cdot a^{n-1} =F\cdot a^{n-2}$ follows.
  \qed
  \end{proof}

  The following lemma shows the lower bound $2^{n-2}$
  on the state complexity of the reverse of binary \R-trivial regular languages.

  \begin{lemma}\label{le:binaryRtriv_lower}
    For every $n\ge3$,
    there exists a binary \R-trivial regular language $L$ with $\sc(L)=n$
    such that $\sc(L^R)\ge 2^{n-2}$.
  \end{lemma}
  \begin{proof}
    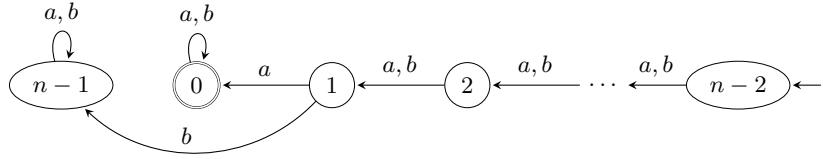
\begin{figure}[t]
      \centering
      \begin{tikzpicture}[->,>=stealth,shorten >=1pt,auto,node distance=1.8cm,
        state/.style={ellipse,minimum size=6mm,very thin,draw=black,initial text=}]
        \node[state,accepting] (0) {$0$};
        \node[state] (1) [right of=0] {$1$};
        \node[state] (2) [right of=1] {$2$};
        \node[     ] (4) [right of=2] {$\ldots$};
        \node[state,initial right] (5) [right of=4] {$n-2$};
        \node[state] (d) [left of=0] {$n-1$};
        \path
          (0) edge[loop above] node {$a,b$} (0)
          (d) edge[loop above] node {$a,b$} (d)
          (1) edge node[above] {$a$} (0)
          (1) edge[bend left=45] node[above] {$b$\quad\,} (d)
          (2) edge node[above] {$a,b$} (1)
          (5) edge node[above] {$a,b$} (4)
          (4) edge node[above] {$a,b$} (2) ;
      \end{tikzpicture}
      \caption{A binary \R-trivial regular language meeting the bound $2^{n-2}$ for the reverse.}
      \label{fi:binaryRtriv}
    \end{figure}
    Consider the language $L$
    accepted by the partially ordered binary \mbox{$n$-state} DFA $M$
    depicted in Fig.~\ref{fi:binaryRtriv}.
    We show that 
    each subset of $\{0,1,\ldots,n-2\}$ containing $0$
    is reachable in the subset automaton of the NFA $M^R$.
    The proof is by induction on the size of subsets.
    The subset $\{0\}$ is the initial state of the subset automaton.
    Each subset $\{0,i_1,i_2,\ldots,i_k\}$ of size $k+1$ 
    with $1\le i_1 < i_2 < \cdots <i_k \le n-2$
    is reached from the subset $\{0,i_2-i_1,\ldots,i_k-i_1\}$ of size $k$
    by the string $ab^{i_1-1}$. 
  \qed
  \end{proof}

  Using a computer program 
  we have computed a few tight bounds 
  summarized in Table~\ref{tab}.
  The bound $2^{n-2}+(n-1)$ is met by a DFA for $L$ with $\sc(L)=n$ if $n\le 6$,
  but not if $n=7$. 
  In addition, more than $2^{n-2}$ states are reachable if $n\le 7$,
  but not if $n=8$.
  By Lemma~\ref{lem:complement}, this means that for $n=8$, the worst-case minimal partially ordered DFA has a dead state and so does its complement.
  It is worth mentioning that the witness languages are even \J-trivial, hence these tight upper bounds also apply to binary \J-trivial regular languages discussed in the next section.
  \begin{table}[h]
    \centering
      \begin{tabular}{c|c|c||c|c|l}
                     & \multicolumn{2}{c||}{Worst-case $\sc(L^R)$} &  &  & \\
                     & \multicolumn{2}{c||}{where DFA for $L$ is}  &  &  & \\
        $n=$         & without    & with                           & Upper bound   & Lower bound & \\
        $\sc(L)$     & dead state & dead state                     & $2^{n-2}+n-1$ & $2^{n-2}$   & \multicolumn{1}{c}{Witness}\\
        \hline
        1 & 1        & 1        & 1/2       & 1/2     & \\
        2 & 2        & 2        & 2         & 1       & $L_2=a^*b(a+b)^*$ \\
        3 & 4        & 4        & 4         & 2       & $L_3=b^*+b^*aL_2$ \\
        4 & 7        & 7        & 7         & 4       & $L_4=b^*aL_3$ \\
        5 & 12       & 12       & 12        & 8       & $L_5=b^*a(aL_3+bL_2)$ \\
        6 & 21       & 21       & 21        & 16      & $L_6=b^*a(b^*a + L_5)$ \\
        7 & {\bf 34} & 34       & {\bf 38}  & 32      & $b^*ab^*a(a+b)(\eps+aL_3+bL_2)$ \\
        8 & 55       & {\bf 64} & 71        & {\bf 64}&
      \end{tabular}
      \smallskip
      \caption{Tight bounds for the reverse of binary \R-trivial regular languages.}
      \label{tab}
  \end{table}
  
  We now prove that for $n\ge 8$, the upper bound is $2^{n-2}$ for binary languages. 

  \begin{lemma}\label{lem:eight}
    Let $n\ge8$ and let $L$ be a binary \R-trivial regular language with $\sc(L)=n$.
    Then $\sc(L^R)\le 2^{n-2}$ and the bound is tight.
  \end{lemma}
  \begin{proof}
    Consider a minimal partially ordered $n$-state DFA $M$
    over a binary alphabet $\{a,b\}$.
    By definition, each maximal state of $M$ has self-loops under both letters $a$ and $b$, hence there are at most two nonequivalent maximal states in $M$.
    
    If there are two maximal states,
    then one of them is accepting and the other one is the dead state.
    The accepting state appears in all reachable subsets
    of the subset automaton of the NFA $M^R$,
    while the dead state appears in no reachable subset.
    Hence the number of reachable subsets is bounded by $2^{n-2}$.
    
    It remains to prove that $2^{n-2}$ is also the bound for $M$ with only one maximal state.
    If the only maximal state is the dead state, we take the complement that has the same state complexity by Lemma~\ref{lem:complement}
    and has no dead state. 
    Thus, assume that $M$ has a single maximal state, $n$, which is accepting.
    Note that if a minimal binary partially ordered DFA has at least four states,
    there is a path of length two in the automaton.
    Consider three last states of such a longest path, say $(n-2) \to (n-1) \to n$.
    In particular, there is no longer path from $n-2$ to $n$.  
    Note also that $n-1$ is not accepting, otherwise it is equivalent to $n$.
    As in the proof of Lemma~\ref{lem:binaryRtriv_upper},
    we can show that to reach the upper bound,
    the situation between states $n-1$ and $n$ must be as depicted in Fig.~\ref{fig:bound4}.
    \begin{figure}[t]
      \centering
      \begin{tikzpicture}[->,>=stealth,shorten >=1pt,auto,node distance=3.2cm,
        state/.style={ellipse,minimum size=6mm, very thin,draw=black,initial text=}]
        \node[state,accepting] (1) {$n$};
        \node[state] (2) [left of=1] {$n-1$};
        \node[state] (3) [left of=2] {$n-2$};
        \path
          (1) edge[loop above] node {$a,b$} (1)
          (2) edge node[above] {$b$} (1)
          (2) edge[loop above] node {$a$} (2)
          (3) edge node[above] {$x$} (2);
      \end{tikzpicture}
      \caption{Path of length two, where $x\in\{a,b\}$.}
      \label{fig:bound4}
    \end{figure}
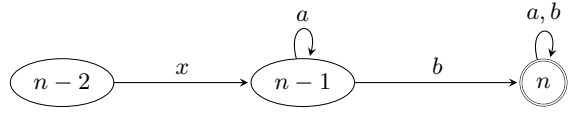
    
    We now compute the number of reachable sets in the subset automaton of the NFA $M^R$
    containing $n$ and $n-1$, but not $n-2$.
    Recall from the proof of Lemma~\ref{lem:binaryRtriv_upper} that $F \cdot a^k$, $k\ge0$,
    reaches at most $n-1$ different subsets. 

    If both $a,b$ go from state $n-2$ to state $n-1$,
    then there are at most $n-1$ subsets in the subset automaton of the NFA $M^R$ containing $n$ and $n-1$ and not $n-2$,
    namely $F\cdot a^k \cdot b$ with $k\ge 0$.
    
    If $x=a$, cf. Fig.~\ref{fig:bound4}, we have the following cases:
    (i) $b$ goes to $n$,
    (ii) $b$ goes to another state $p\notin\{n,n-1,n-2\}$ (the case $p=n-1$ is discussed above), or
    (iii) $b$ is a self-loop in $n-2$.

    In the first case, there is no subset containing $n$ and $n-1$ and not $n-2$ reachable in the subset automaton of $M^R$ 
    because $n-1$ is introduced by $b$, which also introduces $n-2$.

    In the second case, 
    $p$ must go to $n$ (and only to $n$ or $p$) 
    because $M$ has only one maximal state, $n$, and there is no longer path from $n-2$ to $n$.
    If $p$ goes to $n$ under $a,b$,
    then $p$ appears in all subsets containing $n$ and $n-1$, 
    hence only $F\cdot b$ contains $n$ and $n-1$ and not $n-2$.
    If $p$ goes to $n$ under $a$ and $b$ is a self-loop in $p$,
    then there are at most $(n-2)$ subsets containing $n$ and $n-1$ and not $n-2$, namely $F \cdot b \cdot b^k$ with $k\ge0$,
    computed similarly as in the proof of Lemma~\ref{lem:binaryRtriv_upper}.
    If $p$ goes to $n$ under $b$ and $a$ is a self-loop in $p$, then $p$ is equivalent to $n-1$ (if $p$ is non-accepting) or to $n$ (if $p$ is accepting), hence it is not possible.

    In the third case,
    all subsets reachable in the subset automaton of $M^R$ containing $n$ and $n-1$ and not $n-2$ are $F \cdot a^k \cdot b \cdot b^\ell$ with $k,\ell\ge0$. There are at most $(n-2)^2$ such subsets (at most $n-2$ nonempty subsets $F\cdot a^k$ in this case).

    If $x=b$, we have the following cases:
    (i) $a$ goes to $n$,
    (ii) $a$ goes to another state $p\notin\{n,n-1,n-2\}$, or
    (iii) $a$ is a self-loop in $n-2$.

    In the first case, there at most $n-1$ subsets containing $n$ and $n-1$ and not $n-2$, namely $F\cdot a^k \cdot b$ with $k\ge 0$.

    In the second case,
    $p$ must again go to $n$ (and only to $n$ or $p$) for the same reason as above.
    If $p$ goes to $n$ under $a,b$,
    then $p$ appears in all subsets containing $n$ and $n-1$, 
    hence at most $n-1$ subsets, $F\cdot a^k \cdot b$ with $k\ge 0$, contain $n$ and $n-1$ and not $n-2$.
    If $p$ goes to $n$ under $a$ and $b$ is a self-loop in $p$,
    then there are at most $2n-3$ subsets containing $n$ and $n-1$ and not $n-2$,
    namely $n-1$ subsets $F\cdot a^k\cdot b$ and $n-2$ subsets $F\cdot b \cdot b^k\cdot a$ with $k\ge 0$.
    If $p$ goes to $n$ under $b$ and $a$ is a self-loop in $p$,
    then $p$ is equivalent to $n-1$ (or to $n$, see above).

    In the third case, all subsets containing $n$ and $n-1$ and not $n-2$ are reachable only by strings with one $b$, i.e., the reachable subsets are $F \cdot a^k \cdot b \cdot a^\ell$ with $k,\ell \ge 0$. Their number is at most $(n-2)^2$ (at most $n-2$ subsets $F\cdot a^k$ in this case).

    By the proof of Lemma~\ref{lem:binaryRtriv_upper}, there are at most $2^{n-2}$ reachable sets in the subset automaton of $M^R$ containing $n$ and $n-1$, and at most $n-1$ reachable states not containing $n-1$. Thus, for $n\ge 4$ and $M$ with no dead state, the subset automaton of $M^R$ has at most $2^{n-3} + \min(\max(2n-3,(n-2)^2),2^{n-3}) + (n-1)$ reachable states (those containing $n,n-1,n-2$, those containing $n,n-1$ and not $n-2$, and those containing $n$ and not $n-1$, respectively),
    which is less than $2^{n-2}$ for $n\ge 9$.
    For $n=8$ is the bound given by computation (Table~\ref{tab}).
  \qed
  \end{proof}

  Denote by $f_k(n)$ the state complexity function of the reverse
  on binary \R-trivial regular languages over a $k$-element alphabet defined by
  \begin{equation*}\label{sc}
    f_k(n) = \max \makeset{\sc(L^R)}{L\subseteq\Sigma^*, |\Sigma|=k,
             L \text{ is \R-trivial regular, and }
             \sc(L)=n}.
  \end{equation*}
  Using this notation, 
  we can summarize our results
  in the following theorem.

  \begin{theorem}\label{thm1}
    Let $n\ge1$ and let $f_k(n)$ be the state complexity of
    the reverse
    on \R-trivial regular languages over a $k$-element alphabet.
    Then
    \begin{align*}
      & f_1(n) = n, \\
      & f_2(n)= \left\{\begin{array}{ll}
                    1,            & \text{ if $n=1$,} \\
                    2^{n-2}+n-1,  & \text{ if $2\le n\le 6$,} \\         
                    34,           & \text{ if $n=7$,} \\
                    2^{n-2},      & \text{ otherwise,} \\
                          \end{array}\right.\\
      &  f_3(n) = f_k(n) = 2^{n-1}, \text{ for every } k\ge3.
    \end{align*}
  \end{theorem}  
  \begin{proof}
    Since the reverse of every unary language is the same language,
    we have $f_1(n)=n$.
    The upper bounds on $f_2$ are given by Lemmas~\ref{lem:binaryRtriv_upper} and~\ref{lem:eight} and by 
    our calculations in the case of $n=7$.
    The lower bounds in the case of $1\le n\le 7$
    also follow from the calculations,
    while the case of $n\ge8$ is covered by Lemma~\ref{le:binaryRtriv_lower}.
    The result for $f_3$ is from Corollary~\ref{cor:bound}.
    Since adding  new letters to the ternary witness automata does not change
    the proofs  of reachability and distinguishability
    in the ternary case,
    the upper bound is tight for every $k\ge 3$.
  \qed
  \end{proof}

\section{\J-trivial regular languages}\label{***Jtriv}
  Every \J-trivial regular language is also \R-trivial, hence the previous bounds apply. To prove the results of this section, we first define {\em Simon's condition\/} on \R-trivial regular languages to be \J-trivial. 
  
  Let $M=(Q,\Sigma,\delta,q_0,F)$ be a DFA. It can be turned into a directed graph $G(M)$ with the set of vertices $Q$, where a pair $(p,q)\in Q\times Q$ is an edge in $G(M)$ if there is a transition from $p$ to $q$ in $M$. For $\Gamma\subseteq \Sigma$, we define the directed graph $G(M,\Gamma)$ with the set of vertices $Q$ by considering only those transitions that correspond to letters in $\Gamma$.
  
  For a directed graph $G=(V,E)$ and $p\in V$, the set 
  $C(p) = \{q\in V \mid q=p \text{ or there is a directed path from } p \text{ to } q\}$
  is called the component of $p$.

  \begin{definition}[Simon's condition]
    A DFA $M$ with an input alphabet $\Sigma$
    satisfies Simon's condition
    if, for every subset $\Gamma$ of $\Sigma$, each component of $G(M,\Gamma)$ has a unique maximal state.
  \end{definition}

  Simon~\cite{Simon1975} has shown the following result.
  \begin{fact}\label{st}\label{fact:simon}
    An \R-trivial regular language is \J-trivial if and only if its minimal partially ordered DFA satisfies Simon's condition.
  \end{fact}

  Note that it is more efficient to use Trahtman's condition to decide whether an \R-trivial regular language is \J-trivial. For a state $p$, let $\Sigma(p)$ denote the set of letters under which there is a self-loop in $p$. Trahtman has shown that an \R-trivial regular language is \J-trivial if and only if its minimal partially ordered DFA satisfies that, for every state $p$, the connected component of $G(M,\Sigma(p))$ containing $p$ has a unique maximal state, see~\cite{Trahtman2001} for more details.
  
  Using Simon's result we immediately obtain the following lemma.
  \begin{lemma}\label{lem:removing}
    Let  $\Gamma\subseteq\Sigma$.
    If a partially ordered DFA $M$ over $\Sigma$ satisfies Simon's condition, 
    then the DFA $M'$ (not necessarily connected) obtained from $M$
    by removing transitions under letters from $\Gamma$
    also satisfies Simon's condition.
  \end{lemma}
  \begin{proof}
    Let $\Sigma'=\Sigma\setminus\Gamma$.
    By Fact~\ref{fact:simon}, each component of $G(M,\Sigma')$ has a unique maximal state and remains partially ordered.
    \qed
  \end{proof}

  We now prove the main result of this section.
  \begin{theorem}\label{thm:boundJtriv}
    At least $n-1$ letters are necessary for a \J-trivial regular language of the state complexity $n$ to reach the state complexity $2^{n-1}$ in the reverse.
  \end{theorem}
  \begin{proof}
    We prove by induction on the number of states
    that every partially ordered DFA $M$ satisfying Simon's condition with $n\ge 3$ states and at most $n-2$ letters
    has less than $2^{n-1}$ subsets reachable in the subset automaton of $M^R$. 
 
    The basis for $n=3$ holds since the automaton is over a unary alphabet, 
    which means that the set $\{F\cdot a^k \mid k\ge 0\}$ has at most three elements (cf. the proof of Lemma~\ref{lem:binaryRtriv_upper}).
    
    Assume that for some $k\ge 3$ the claim holds for every partially ordered DFA satisfying Simon's condition with at most $k$ states and $k-2$ letters. Let $M=(Q,\Sigma,\delta,q_0,F)$ be a partially ordered DFA satisfying Simon's condition with $|Q|=k+1$ states and $|\Sigma|<|Q|-1$ letters. We prove that less than $2^{|Q|-1}$ subsets are reachable in the subset automaton of the NFA $M^R$. To do this, we show that reachability of $2^{|Q|-1}$ subsets in the subset automaton of $M^R$ implies the existence of a partially ordered DFA $M''=(Q'',\Sigma'',\delta'',q_0'',F'')$ satisfying Simon's condition with $|\Sigma''|<|Q''|-1$ letters, $|Q''|\le k$ states and $2^{|Q''|-1}$ reachable subsets in the subset automaton of the NFA $M''^R$. However, by the induction hypothesis, the number of reachable subsets in the subset automaton of $M''^R$ is less than $2^{|Q''|-1}$, which means that the assumption of $2^{|Q|-1}$ reachable subsets in the subset automaton of $M^R$ cannot hold.
    
    We may assume that $M$ is connected and has no equivalent states, since any two equivalent states of $M$ appear in the same sets in the subset automaton of $M^R$, which implies reachability of less than $2^{|Q|-1}$ subsets in the subset automaton of $M^R$. Similarly for two or more connected components. We may also assume that the unique maximal state of $M$, denoted by $n$, is accepting. Indeed, a subset $X\subseteq Q$ is reachable in the subset automaton of the reverse of $M$ if and only if the set $Q\setminus X$ is reachable in the subset automaton of the reverse of the complement of $M$.
    
    To construct $M''$, we first define nonempty sets $S\subseteq Q\setminus F$ and $\Gamma\subseteq\Sigma$ such that $|S|\le|\Gamma|$ and use them to construct the (not necessarily connected) partially ordered DFA $M''$ from $M$ by removing state $n$ and all transitions labeled by letters from $\Gamma$ and joining all states of $S$ into a single state. We show that $M''$ satisfies Simon's condition and that it has $2^{|Q''|-1}$ reachable subsets in the subset automaton of the reverse. Since $|\Sigma|<|Q|-1$ and $|S|\le|\Gamma|$, we obtain that $|\Sigma''|=|\Sigma|-|\Gamma| < |Q|-|S|-1=|Q''|-1<k$ and induction applies.
    
    To construct the sets $S$ and $\Gamma$, 
    let $R=\{q\in Q\setminus\{n\} \mid \delta(q,a)=n, a\in\Sigma\}$ denote the set of all states different from $n$ with a transition to $n$,
    and let $\Gamma=\{a\in\Sigma \mid \delta(R,a)\cap\{n\}\neq\emptyset\}$ denote the set of letters connecting states of $R$ with state $n$.
    Note that $R$ and $\Gamma$ are nonempty.
    Let $M'$ be the $k$-state subautomaton of $M$ obtained by removing state $n$
    and all transitions labeled by letters from $\Gamma$.
    By Lemma~\ref{lem:removing}, $M'$ satisfies Simon's condition.
    
    Let $\max(R)$ denote the set of states of $R$ that are maximal in $M'$.
    For a state $p$ in $\max(R)$, let $C_{p}$ denote the connected component of $G(M')$ containing $p$,
    and let $\Sigma_p=\{a\in\Sigma \mid \delta(p,a)=n\}\subseteq\Gamma$ denote the set of labels connecting $p$ to $n$, see Fig.~\ref{fig:tree} for illustration.
    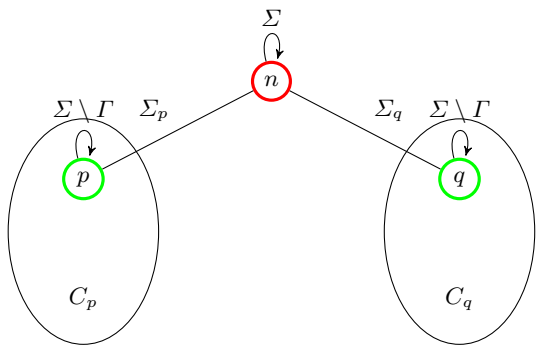
\begin{figure}[t]
      \tikzset{ 
        treenode/.style = {align=center, inner sep=0pt}, 
        ac/.style = {treenode, circle, black, draw=red,   text width=1.5em, very thick}, 
        na/.style = {treenode, circle, black,   draw=green, text width=1.5em, very thick} 
      }
      \centering
      \begin{tikzpicture}[-,>=stealth',level/.style={sibling distance = 5cm/#1, level distance = 1.3cm}]
        \node [ac] (1) {$n$}
          child{ node [na] (2) {$p$} edge from parent node[above left]  {$\Sigma_{p}$} }
          child{ node [na] (3) {$q$} edge from parent node[above right] {$\Sigma_{q}$} } ;
        \path
          (1) edge[loop above] node {$\Sigma$} (1)
          (2) edge[loop above] node {$\Sigma\setminus\Gamma$} (2)
          (3) edge[loop above] node {$\Sigma\setminus\Gamma$} (3);
        \draw (-2.5,-2) ellipse (1cm and 1.5cm);
        \draw (-2.5,-2.9) node {$C_{p}$};
        \draw (2.5,-2) ellipse (1cm and 1.5cm);
        \draw (2.5,-2.9) node {$C_{q}$};
      \end{tikzpicture}
      \caption{The partially ordered DFA $M'$; $\Gamma=\Sigma_p\cup\Sigma_q$.}
      \label{fig:tree}
    \end{figure}
    Note that $C_p$ and $C_q$ are not connected, for $p\neq q$, otherwise $p$ and $q$ are two maximal states of the connected component containing $C_p\cup C_q$. Next, we show that for every letter $a$ in $\Gamma$, there exists a state $s$ in $\max(R)$ such that $s$ goes to $n$ under $a$.
    Let $a$ be a letter from $\Gamma$ and let $r$ be a state in $R$ with an $a$-transition to $n$ such that no other state reachable from $r$ in $M$ goes to $n$ under $a$. If $r$ is not in $\max(R)$, there is a state $t$ in $\max(R)$ such that $r$ belongs to $C_t$. But then there are two maximal states in the component containing $r$ in the graph $G(M,(\Sigma\setminus\Gamma)\cup\{a\})$, namely $n$ and the one reachable from $t$ by letters from $(\Sigma\setminus\Gamma)\cup\{a\}$. Thus, $\Gamma = \bigcup_{p\in\max(R)} \Sigma_p$.
    Note that all states of $\max(R)$ are non-accepting; if a state $s$ of $\max(R)$ is accepting, then, by the assumption, subset $\{n\}$ is reachable in the subset automaton of $M^R$. This requires to eliminate state $s$ from the initial state $F$. However, it can be done only by a letter from $\Gamma$, which always introduces another state from $\max(R)$.
    
    We now prove that $|\Gamma|\ge |\max(R)|$, i.e., for every state $p$ in $\max(R)$,
    there exists a letter $\sigma_p$ in $\Sigma_p$
    that does not appear in $\Sigma_q$ for any other state $q$ in $\max(R)$.
    For the sake of contradiction,
    assume that there is a state $p$ in $\max(R)$ with $\Sigma_{p}\subseteq\bigcup_{q\in\max(R)}^{q\neq p}\Sigma_{q}$.
%
    Since all subsets containing $n$ and $p$ and not any $q$ from $\max(R)$
    different from $p$ are reachable in the subset automaton of $M^R$,
    state $p$ is introduced to the subset from $n$
    by a transition under a letter from $\Sigma_p$
    which also introduces a state $q\neq p$ into that subset.
    Since $\Gamma=\bigcup_{q\in\max(R)}^{q\neq p}\Sigma_{q}$, 
    any attempt to eliminate state $q$ results in the introduction of a state $q'$ different from $p$, which is a contradiction.
%
    

    Let $S=\max(R)$. Then $|\Gamma|\ge |S|$ as required.
    Recall that the states of $S$ are maximal in $M'$ and non-accepting.
    Thus, they do not appear in any reachable subset of the subset automaton of the reverse of $M'$.
    Construct the DFA $M''$ from $M'$
    by joining all states of $S$ into one state.
    Then the subset automaton of the reverse of $M''$
    has the same number of reachable subsets
    as the subset automaton of the reverse of $M'$,
    and $M''$ satisfies Simon's condition.

    Finally, we show that $M'$ (hence also $M''$) has $2^{|Q''|-1}$ reachable subsets in the subset automaton of the reverse.
    Since $n$ is accepting, each set $X$ containing $n$ and nothing from $S$ is reachable in the subset automaton of $M^R$ only by symbols from $\Sigma''=\Sigma\setminus\Gamma$ (otherwise a symbol from $S$ is introduced and we cannot get rid of states of $S$ anymore),
    hence the set $X\setminus\{n\}$ is reachable in the subset automaton of $M'^R$.
    As there are $2^{|Q|-(|S|+1)}=2^{|Q''|-1}$ such sets, $M''$ has $2^{|Q''|-1}$ reachable subsets in the subset automaton of the reverse.
    This leads to the contradiction explained above and completes the proof.
  \qed\end{proof}
  
  Using this result, we can prove the tight upper bound 
  on the state complexity of the reverse 
  for \J-trivial regular languages over an $(n-2)$-element alphabet.

  \begin{theorem}\label{thm3}
    Let $n\ge3$ and let $L$ be a \J-trivial regular language
    over an $(n-2)$-element alphabet with  $\sc(L)=n$.
    Then $\sc(L^R)\le 2^{n-1}-1$ and the bound is tight.
  \end{theorem}
  \begin{proof}
    The upper bound follows from Theorem~\ref{thm:boundJtriv}.
    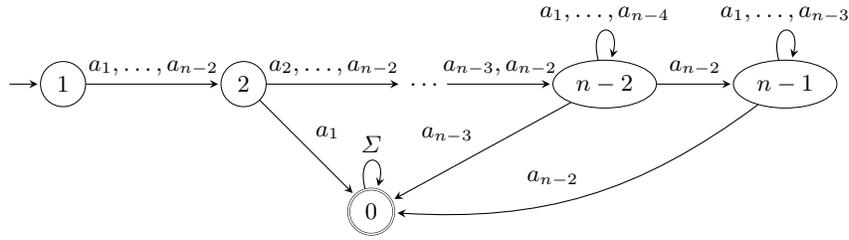
\begin{figure}[t]
      \centering
      \begin{tikzpicture}[->,>=stealth,shorten >=1pt,auto,node distance=2.4cm,
        state/.style={ellipse,minimum size=6mm, very thin,draw=black,initial text=}]
        \node[state,initial] (2) {$1$};
        \node[state] (3) [right of=2] {$2$};
        \node[     ] (4) [right of=3] {$\ldots$};
        \node[state] (5) [right of=4] {$n-2$};
        \node[state] (6) [right of=5] {$n-1$};
        \node[state,accepting] (d) [below right of=3] {$0$};
        \path
          (2) edge node {$a_1,\ldots,a_{n-2}$} (3)
          (3) edge node {$a_1$} (d)
          (3) edge node {$a_2,\ldots,a_{n-2}$} (4)
          (4) edge node {$a_{n-3},a_{n-2}$} (5)
          (5) edge[loop above] node {$a_1,\ldots,a_{n-4}$} (5)
          (5) edge node[above left] {$a_{n-3}$} (d)
          (5) edge node {$a_{n-2}$} (6)
          (6) edge[loop above] node {$a_1,\ldots,a_{n-3}$} (6)
          (6) edge[bend left=20] node[above left] {$a_{n-2}$} (d)
          (d) edge[loop above] node {$\Sigma$} (d) ;
      \end{tikzpicture}
      \caption{The witness minimal partially ordered DFA $M$ satisfying Simon's condition.}
      \label{fig:n2witness}
    \end{figure}
    Thus, to prove tightness, we consider 
    the \J-trivial regular language accepted by the minimal DFA 
    $M = (\{0,1,\ldots,n-1\},\{a_1,\ldots,a_{n-2}\},\delta,1,\{0\})$
    depicted in Fig.~\ref{fig:n2witness}.
    The transitions under a letter $a_j$ in $\Sigma$ are defined by
    \begin{eqnarray*} 
    \delta(i,a_j)= \left\{\begin{array}{ll}
                    i+1, & \text{if $i\le j\le n-2$,} \\
                    0,   & \text{if $i = j-1$,} \\
                    i,   & \text{otherwise.}
                  \end{array}\right.
    \end{eqnarray*}
    The initial state of the subset automaton of the NFA $M^R$ is the set $\{0\}$ and, for $1\le k\le n-2$, 
    every $(k+1)$-element set $\{0,i_1,i_2,\ldots,i_k\}$ 
    with $2\le i_1<i_2<\cdots<i_k\le n-1$ 
    is reached from the $k$-element set $\{0,i_2,\ldots,i_k\}$ by $a_{i_1-1}$.
    This gives $2^{n-2}$ reachable states (those containing 0 and not 1).
    Note that the set $\{0,1\}$ is not reachable,
    but all subsets of the state set of $M$ of cardinality at least three containing $0$ and $1$
    are reachable
    since every set $\{0 ,1,i_1,\ldots,i_k\}$ is reached from the set
    $\{0,2,i_2,\ldots,i_k\}$ by letter $a_{i_1-1}$.
    \qed
  \end{proof}

  Lemma~\ref{lem:eight} also gives the upper bound for binary \J-trivial regular languages. The witness languages in Table~\ref{tab} are \J-trivial. For $n\ge 8$, we need a dead state to reach the upper bound $2^{n-2}$ and our witness automata have a dead state (Fig.~\ref{fi:binaryRtriv}), hence the language is not \J-trivial.
  \begin{corollary}
    Let $L$ be a binary \J-trivial regular language with $\sc(L)=n$, 
    where $n\ge 4$, then $\sc(L^R)\le 2^{n-3} + \min(\max(2n-3,(n-2)^2),2^{n-3}) + (n-1)$. 
    A few tight upper bounds for $2\le n\le 7$ are given in Table~\ref{tab}.
    \qed
  \end{corollary}
  
  Concerning the lower bound state complexity, it was shown in~\cite{CampeanuCSY99} that there are finite binary languages whose reverse have a blow-up of $3 \cdot 2^{\frac{n}{2}-1} - 1$, for $n$ even, and of $2^{\frac{n+1}{2}} - 1$, for $n$ odd. Since every finite language is \J-trivial, we obtain at least these lower bounds for binary \J-trivial regular languages.

\section{Conclusions}\label{***conclu}
  We have presented a characterization of tight bounds on the state complexity of the reverse for \R-trivial regular languages depending not only on the state complexity of the language, but also on the size of its alphabet. As a consequence, this characterization also gives upper bounds for \J-trivial regular languages, but they are not reachable for languages of the state complexity $n$ over an $(n-k)$-element alphabet, for $2\le k\le n-3$. We have further shown tight bounds for \J-trivial regular languages over $(n-1)$- and $(n-2)$-element alphabets, but (except for a few examples for binary \J-trivial regular languages) the problem of the tight bounds for \J-trivial regular languages over an alphabet of a lower cardinality is open.

\subsubsection*{Acknowledgements.}
  The authors gratefully acknowledge comments and suggestions of anonymous referees.

\bibliographystyle{splncs}
\bibliography{dcfs2013}

\end{document}